\documentclass[11pt]{llncs}

\usepackage{fullpage}
\usepackage{amsmath}
\usepackage{graphicx}

\newcommand{\rank}{\ensuremath{\textit{rank}}}
\newcommand{\select}{\ensuremath{\textit{select}}}

\newcommand{\SA}{\ensuremath{A}}
\newcommand{\LCP}{\ensuremath{H}}

\title{Wee LCP}
\author{Johannes Fischer}
\institute{Karlsruhe Institute of Technology, 76128 Karlsruhe, Germany\\
\email{fischer@informatik.uni-tuebingen.de}}

\begin{document}
\maketitle

\begin{abstract}
We prove that longest common prefix (LCP) information can be stored in much less space than previously known. More precisely, we show that in the presence of the text and the suffix array, $o(n)$ additional bits are sufficient to answer LCP-queries asymptotically in the same time that is needed to retrieve an entry from the suffix array. This yields the smallest compressed suffix tree with sub-logarithmic navigation time.
\end{abstract}

\section{Introduction}
Augmenting the suffix-array \cite{gonnet92new,manber93suffix} with an additional array holding the lengths of \emph{longest common prefixes} drastically increases its functionality \cite{manber93suffix,abouelhoda04replacing,cole06suffix}. Stored as a plain array, this so-called \emph{LCP-array} occupies $n\lceil\log n\rceil$ bits for a text of length $n$. Sadakane \cite{sadakane07compressed} shows that less space is actually needed, namely $2n+o(n)$ bits if the suffix array is available at lookup time (see Sect.\ \ref{sect:lcp_succinct} for more details). Due to the $2n$-bit term, this data structure is nevertheless \emph{incompressible}, even for highly regular texts. As a simple example, suppose the text consists of only a's. Then the suffix array can be compressed to almost negligible space \cite{sadakane03new,grossi03high,ferragina05indexing}, while Sadakane's LCP-array cannot.

Text regularity is usually measured in \emph{$k$-th order empirical entropy} $H_k$ \cite{manzini01analysis}. We have $0 \le H_k \le \log \sigma$ for a text on an alphabet of size $\sigma$, with $H_k$ being ``small'' for compressible texts. In Sect.\ \ref{sect:wee_lcp} of this article, we prove that the LCP-array can be stored in $O\left(\frac{n}{\log\log n}\right)$ or $nH_k+o(n)$ bits (depending on how the text itself is stored), while giving access to an arbitrary LCP-value in time $O(\log^\delta n)$ (arbitrary constant $0 < \delta < 1$). This should be compared to other compressed or sampled variants of the LCP-array. We are aware of three such methods:
\begin{enumerate}
\item Russo et al.\ \cite{russo08fully} achieve $nH_k + o(n)$ space, but retrieval time at least $O(\log^{1+\epsilon}n)$, hence super-logarithmic (arbitrary constant $0 < \epsilon < 1$).
\item Fischer et al.\ \cite{fischer09faster} achieve $nH_k(\log\frac{1}{H_k} + O(1))$ bits, with retrieval time $O(\log^{\beta}n)$, again for any constant $0 < \beta < 1$. Although the space vanishes if $H_k$ does, it is worse than our data structure.
\item K\"arkk\"ainen et al.\ \cite[Lemma 3]{kaerkkaeinen09permuted} also employ the idea of ``sampling'' the LCP-array, but achieve only \emph{amortized} time bounds. Allowing the same space as for our data structure ($O(\frac{n}{\log\log n})$ bits on top of the suffix array and the text), one would have to choose $q=\log n \log\log n$ in their scheme, yielding super-logarithmic $O(\log n \log\log n)$ amortized retrieval time.
\end{enumerate}

Finally, in Sect.\ \ref{sect:cst}, we apply our new representation of the LCP-array to suffix trees. This yields the first compressed suffix tree with $O(nH_k)$ bits of space and sub-logarithmic navigation-time for almost all operations.

\section{Definitions and Previous Results}

\label{sect:definitions}
This section sketches some \emph{known} data structures that we are going to make use of. Throughout this article, we use 
the standard \emph{word-RAM} model of computation, in which we have a computer with word-width $w$, where $\log n = O(w)$. Fundamental arithmetic operations (addition, shifts, multiplication, \dots) on $w$-bit wide words can be computed in $O(1)$ time.

\subsection{Rank and Select on Binary Strings}
\label{sect:rank}
Consider a \emph{bit-string} $S[1,n]$ of length $n$. We define the 
fundamental \emph{rank}- and \emph{select}-operations on $S$ as follows: 
$\rank_1(S,i)$ gives the number of 1's in the prefix $S[1,i]$, and 
$\select_1(S,i)$ gives the position of the $i$'th 1 in $S$, reading $S$ from 
left to right ($1 \le i \le n$). Operations $\rank_0(S,i)$ and $\select_0(S,i)$ are 
defined similarly for 0-bits. There are data structures of size $O(\frac{n\log\log n}{\log n})$ bits in addition to $S$ that support $O(1)$-rank- and select-operations, respectively \cite{jacobson89space,golynski07optimal}. For an easily accessible exposition of these techniques, we refer the reader to the survey by Navarro and M\"akinen \cite[Sect.\ 6.1]{navarro07compressed}.

\subsection{Suffix- and LCP-Arrays}
\label{sect:def_sa}
The \emph{suffix array} \cite{gonnet92new,manber93suffix} for a given text $T$ of length $n$ is an array $\SA[1,n]$ of integers s.t.\ $T_{\SA[i]..n} < T_{\SA[i+1]..n}$ for all $1 \le i < n$; i.e., $\SA$ describes the lexicographic order of $T$'s suffixes by ``enumerating'' them from the lexicographically smallest to the largest. It follows from this definition that $\SA$ is a permutation of the numbers $[1,n]$. Take, for example, the string $T=\mathrm{CACAACCAC\$}$. Then $\SA=[10,4,8,2,5,9,3,7,1,6]$. Note that the suffix array is actually a ``left-to-right'' (i.e., alphabetical) enumeration of the leaves in the suffix tree \cite[Part II]{gusfield97algorithms} for $T\$$. As $\SA$ stores $n$ integers from the range $[1,n]$, it takes $n$ words (or $n\lceil\log n\rceil$ bits) to store $\SA$ in uncompressed form. However, there are also different variants of \emph{compressed} suffix arrays; see again the survey by Navarro and M\"akinen \cite{navarro07compressed} for an overview of this field. In all cases, the time to access an arbitrary entry $\SA[i]$ rises to $\omega(1)$; we denote this time by $t_\SA$. All current compressed suffix arrays have $t_\SA=\Omega(\log^\epsilon n)$ in the worst case (arbitrary constant $0< \epsilon \le 1$), and there are indeed ones that achieve this time \cite{grossi03high}.

In the same way as suffix arrays store the leaves of the corresponding suffix tree, the \emph{LCP-array} captures information on the heights of the internal nodes as follows. Array $\LCP[1,n]$ is defined such that $\LCP[i]$ holds the length of the \emph{longest common prefix} of the lexicographically $(i-1)$'st and $i$'th smallest suffixes. In symbols, $\LCP[i]= \max\{k : T_{\SA[i-1]..\SA[i-1]+k-1} = T_{\SA[i]..\SA[i]+k-1}\}$ for all $1 < i \le n$, and $\LCP[1] = 0$. For $T=\mathrm{CACAACCAC\$}$, $\LCP = [0,0,1,2,2,0,1,2,3,1]$. Kasai et al.\ \cite{kasai01linear} gave an algorithm to compute $\LCP$ in $O(n)$ time, and Manzini \cite{manzini04two} adapted this algorithm to work in-place.\footnote{M\"akinen \cite[Fig.\ 3]{maekinen03compact} gives another algorithm to compute $\LCP$ almost in-place.}

\subsection{$2n+o(n)$-Bit Representation of the LCP-Array}
\label{sect:lcp_succinct}
\begin{figure}[t]
\centering
\includegraphics[scale=1]{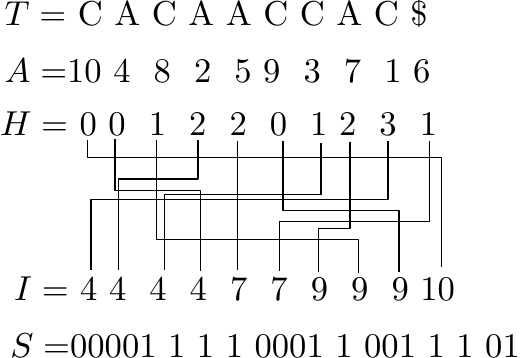}
\caption{Illustration to the succinct representation of the LCP-array.}
\label{fig:succinct-lcp}
\end{figure}

Let us now come to the description of the succinct representation of the LCP-array due to Sadakane \cite{sadakane07compressed}. The key to his result is the fact that the LCP-values cannot decrease too much if listed in the order of the inverse suffix array $A^{-1}$ (defined by $\SA^{-1}[i]=j$ iff $\SA[j]=i$), a fact first proved by Kasai et al.\ \cite[Thm.\ 1]{kasai01linear}:
\begin{proposition}
\label{prop:lcp_cannot_decrease_too_much}
For all $i > 1$, $\LCP[\SA^{-1}[i]] \ge \LCP[\SA^{-1}[i-1]]-1$.
\hfill\qed
\end{proposition}
Because $\LCP[\SA^{-1}[i]] \le n-i+1$ (the LCP-value cannot be longer than the length of the suffix!), this implies that $\LCP[\SA^{-1}[1]]+1, \LCP[\SA^{-1}[2]]+2,\dots,\LCP[\SA^{-1}[n]]+n$ is an increasing sequence of integers in the range $[1,n]$. Now this list can be encoded \emph{differentially}: for all $i=1,2,\dots,n$, subsequently write the difference $I[i] := \LCP[\SA^{-1}[i]]-\LCP[\SA^{-1}[i-1]]+1$ of neighboring elements in unary code $0^{I[i]}1$ into a bit-vector $S$, where we assume $\LCP[\SA^{-1}[-1]]=0$. Here, $0^{x}$ denotes the juxtaposition of $x$ zeros. See also Fig.\ \ref{fig:succinct-lcp}. Combining this with the fact that the LCP-values are all less than $n$, it is obvious that there are at most $n$ zeros and exactly $n$ ones in $S$. Further, if we prepare $S$ for constant-time $\rank_{0}$- and $\select_{1}$-queries, we can retrieve an entry from \LCP\ by
\begin{equation}
\label{eq:lcp_succinct1}
\LCP[i] = \rank_{0}(S, \select_{1}(S, \SA[i]))-\SA[i]\ .
\end{equation}
This is because the $\select_{1}$-statement gives the position where the encoding for $\LCP[\SA[i]]$ ends in $H$, and the $\rank_{0}$-statement counts the sum of the $I[j]$'s for $1\le j \le \SA[i]$. So subtracting the value $\SA[i]$, which has been ``artificially'' added to the LCP-array, yields the correct value. See Fig.\ \ref{fig:succinct-lcp} for an example. Because $\rank_0(H, \select_1(H,x)) = \select_1(H,x)-x$, we can rewrite \eqref{eq:lcp_succinct1} to
\begin{equation}
\label{eq:lcp_succinct2}
\LCP[i] = \select_{1}(S, \SA[i])-2\SA[i]\ ,
\end{equation}
such that only one select-call has to be invoked.

This leads to
\begin{proposition}[Succinct representation of LCP-arrays]
\label{prop:succinct_lcp}
The LCP-array for a text of length $n$ can be stored in $2n+O(\frac{n\log\log n}{\log n})$ bits in addition to the suffix array, while being able to access its elements in time $O(t_\SA)$.
\hfill\qed
\end{proposition}

\section{Less Space, Same Time}
\label{sect:wee_lcp}
The solution from Sect.\ \ref{sect:lcp_succinct} is admittedly elegant, but certainly leaves room for further improvements. Because the bit-vector $S$ stores the LCP-values in \emph{text} order, we first have to convert the position in \SA\ to the corresponding position in the text. Hence, the lookup time to $H$ is dominated by $t_\SA$, the time needed to retrieve an element from from the compressed suffix array. Intuitively, this means that we could take up to $O(t_\SA)$ time to answer the select-query, without slowing down the whole lookup asymptotically. Although this is not exactly what we do, keeping this idea in mind is helpful for the proof of the following theorem.

\begin{theorem}
\label{thm:wee_lcp}
Let $T$ be a text of length $n$ with $O(1)$-access to its characters. Then the LCP-array for $T$ can be stored in $O(\frac{n}{\log\log n}) = o(n)$ bits in addition to $T$ and to the suffix array, while being able to access its elements in $O(t_\SA + \log^\delta n)$ time (arbitrary constant $0<\delta \le 1$).
\end{theorem}
\begin{proof}
We build on the solution from Sect.\ \ref{sect:lcp_succinct}. Let $j = \SA[i]$. From \eqref{eq:lcp_succinct2}, we compute $\LCP[i]$ as $\select_{1}(S, j)-2j$. Computing $\SA[i]$ takes time $t_\SA$. Thus, if we could answer the select-query in the same time (using $O(\frac{n}{\log\log n})$ additional bits), we were done. We now describe a data structure that achieves essentially this. Our description follows in most parts the solution due to Navarro and M\"akinen \cite[Sect.~6.1]{navarro07compressed}, except that it does not store sequence $S$ and the lookup-table on the deepest level.

We divide the $\emph{range}$ of arguments for $\select_1$ into subranges of size $\kappa= \lfloor \log^2{n} \rfloor$, and store in $N[i]$ the answer to $\select_1(S,i\kappa)$. This table $N[1,\lceil \frac{n}{\kappa} \rceil]$ needs $O(\frac{n}{\kappa}\log{n})=O(\frac{n}{\log{n}})$ bits, and divides $S$ into \emph{blocks} of different size, each containing $\kappa$ 1's (apart from the last).

A block is called \emph{long} if it spans more than $\kappa^2=\Theta(\log^4{n})$ positions in $S$, and \emph{short} otherwise. For the long blocks, we store the answers to all $\select_1$-queries explicitly in a table $P$. Because there are at most $\kappa^2$ long blocks, $P$ requires $O\left(\frac{n}{\kappa^2}\times\kappa\times\log{n}\right)=O(n/\log^4{n} \times \log^2{n} \times \log{n})=O\left(n/\log{n}\right)$ bits.

Short blocks contain $\kappa$ 1-bits and span at most $\kappa^2$ positions in $S$. We divide again their range of arguments into sub-ranges of size $\lambda =\lfloor \log^2{\kappa} \rfloor=\Theta(\log^2{\log{n}})$. In $N'[i]$, we store the answer to $\select_1(S,i\lambda)$, this time only relative to the beginning of the block where $i$ occurs. Because the values in $N'$ are in the range $[1,\kappa^2]$, table $N'[1, \lceil \frac{n}{\lambda} \rceil]$ needs $O\left(\frac{n}{\lambda}\times\log{\kappa}\right)=O\left(n/\log{\log{n}}\right)$ bits. Table $N'$ divides the blocks into \emph{miniblocks}, each containing $\lambda$ 1-bits. 

Miniblocks are called \emph{long} if they span more than $s=\log^\delta n$ bits, and \emph{short} otherwise. For long miniblocks, we store again the answers to all $\select$-queries explicitly in a table $P'$, relative to the beginning of the corresponding block. Because the miniblocks are contained in short blocks of length $\le \kappa^2$, the answer to such a $\select$-query takes $O(\log{\kappa})$ bits of space. Thus, the total space for $P'$ is $O(n/s \times \lambda\times\log{\kappa})=O\left(\frac{n\log^3{\log{n}}}{\log^\delta{n}}\right)$ bits. This concludes the description of our data structure for select.

\begin{figure}[t]
\centering
\includegraphics[scale=1]{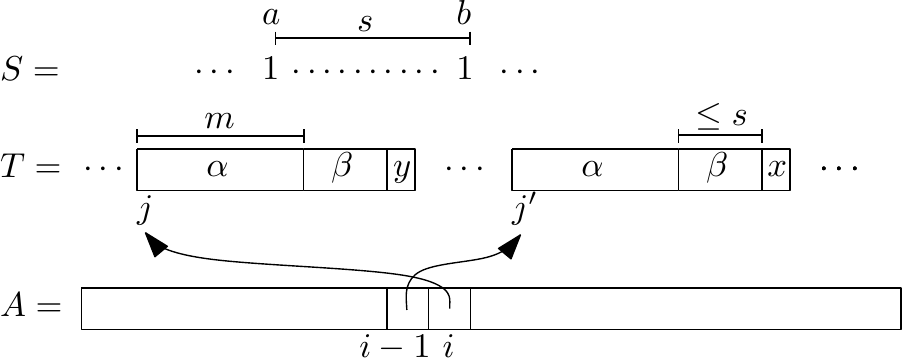}
\caption{Illustration to the proof of Thm.\ \ref{thm:wee_lcp}. We know from the value of $a$ that the first $m=\max(a-2j,0)$ characters of $T_{j\dots n}$ and $T_{j'\dots n}$ match (here with common prefix $\alpha$). At most $s=\log^\delta n$ further characters will match (here $\beta$), until reaching a mismatch character $x <_{\text{lex}} y$.}
\label{fig:ultra-succinct-lcp}
\end{figure}

To answer a query $\select_1(S, j)$, let $a=\select_1(S, \lfloor j/\lambda\rfloor\lambda)$ be the beginning of $j$'s mini-block in $S$. Likewise, compute the beginning of the next mini-block as $b=\select_1(S, \lfloor j/\lambda\rfloor\lambda+\lambda)$. Now if $b-a > s$, then the mini-block where $i$ occurs is long, and we can look up the answer using our precomputed tables $N$, $N'$, $P$, and $P'$. Otherwise, we return the value $a$ as an \emph{approximation} to the actual value of $\select_1(S,j)$.

We now use the text $T$ to compute $\LCP[i]$ in additional $O(\log^\delta n)$ time. To this end, let $j'=\SA[i-1]$ (see also Fig.\ \ref{fig:ultra-succinct-lcp}). The unknown value $\LCP[i]$ equals the length of the longest common prefix of suffixes $T_{j\dots n}$ and $T_{j'\dots n}$, so we need to compare these suffixes. However, we do not have to compare letters from scratch, because we already know that the first $m=\max(a-2j,0)$ characters of these suffixes match. So we start the comparison at $T_{j+m}$ and $T_{j'+m}$, and compare as long as they match. Because $b-a \le s = \log^\delta n$, we will reach a mismatch after at most $s$ character comparisons. Hence, the additional time for the character comparisons is $O(\log^\delta n)$.
\qed
\end{proof}

If the text is not available for $O(1)$ access, we have two options. First, we can always compress it with recent methods \cite{sadakane06squeezing,ferragina07simple,gonzalez06statistical} to $nH_k+o(n)$ space (which is within the space of all compressed suffix arrays), while still guaranteeing $O(1)$ random access to its characters:

\begin{corollary}
\label{cor:wee_lcp1}
The LCP-array for a text of length $n$ can be stored in $nH_k + O\left(\frac{n(k\log\sigma + \log\log n)}{\log_\sigma n} + \frac{n}{\log\log n}\right)$ bits in addition to the suffix array (simultaneously over all $k\in o(\log_\sigma n)$ for alphabet size $\sigma$), while being able to access its elements in time $O(t_\SA + \log^\delta n)$ (arbitrary constant $0<\delta \le 1$).
\qed
\end{corollary}

The second option is to use the compressed suffix array itself to retrieve characters in $O(t_\SA)$ time. This is either already provided by the compressed suffix array \cite{sadakane03new}, or can be simulated \cite{fischer09faster}. This leads to
\begin{corollary}
\label{cor:wee_lcp2}
Let $A$ be a compressed suffix array for a text of length $n$ with access time $t_\SA=O(\log^\epsilon n)$. Then the LCP-array can be stored in $o(n)$ bits in addition to the suffix array, while being able to access its elements in time $O(\log^{\epsilon + \delta} n)$ (arbitrary constant $0<\delta \le 1$).
\qed
\end{corollary}

Note in particular that \emph{all} known compressed suffix arrays have worst-case lookup time $\Theta(\log^\epsilon n)$ at the very best, so the requirements on $t_\SA$ in Cor.\ \ref{cor:wee_lcp2} are no restriction on its applicability. Further, by choosing $\epsilon$ and $\delta$ such that $\epsilon+\delta<1$, the time to access the LCP-values remains sub-logarithmic.

\subsection{Improved Retrieval Time}
Additional time could be saved in Thm.\ \ref{thm:wee_lcp} and Cor.\ \ref{cor:wee_lcp1} by noting that a chunk of $\log_\sigma n$ text characters can be processed in $O(1)$ time in the RAM-model for alphabet size $\sigma$. Hence, when comparing the at most $s=\log^\delta n$ characters from suffixes $T_{j+m\dots n}$ and $T_{j'+m\dots n}$ (end of the proof of Thm.\ \ref{thm:wee_lcp}), this could be done by processing at most $s/\log_\sigma n = \log\sigma \log^{\delta-1} n$ such chunks. This is especially interesting if the alphabet size is small; in particular, if $\sigma = O\left(2^{\left(\log^{1-\delta} n\right)}\right)$, the retrieval time becomes constant.

The same improvement is possible for K\"arkk\"ainen et al.'s solution \cite{kaerkkaeinen09permuted}, resulting in $O(q \log \sigma/ \log n)$ amortized retrieval time in their scheme.

\section{A Small Entropy-Bounded Compressed Suffix Tree}
\label{sect:cst}
The data structure from Thm.\ \ref{thm:wee_lcp} is particularly appealing in the context of compressed suffix trees. Fischer et al.\ \cite{fischer09faster} give a compressed suffix tree that has sub-logarithmic time for almost all navigational operations. It is based on the compressed suffix array due to Grossi et al.\ \cite{grossi03high}, a compressed LCP-array, and data structures for range minimum- and previous/next smaller value-queries (RMQ and PNSV). Its size is $nH_k(2\log\frac{1}{H_k}+\frac{1}{\epsilon} + O(1)) + o(n)$ bits, where the ``ugly'' $nH_k(\log\frac{1}{H_k} + O(1))$-term comes from a compressed form of the LCP-array. If we replace this data structure with our new representation, we get (using Cor.\ \ref{cor:wee_lcp2} for simplicity):

\begin{theorem}
\label{thm:cst1}
A suffix tree can be stored in $(1+\frac{1}{\epsilon})nH_k+o(n)$ bits such that all operations can be computed in sub-logarithmic time (except level ancestor queries, which have an additional $O(\log n)$ penalty).
\end{theorem}
\begin{proof}
The space can be split into $(1+\frac{1}{\epsilon})nH_k + o(n)$ bits from the compressed suffix array \cite{grossi03high}, additional $o(n)$ bits from the LCP-array of Cor.\ \ref{cor:wee_lcp2}, plus $o(n)$ bits for the RMQ- and PNSV-queries. The time bounds are obtained from the third column of Table 1 in \cite{fischer09faster}.
\qed
\end{proof}

Other trade-offs than those in Thm.\ \ref{thm:cst1} are possible, e.g., by taking different suffix arrays, or by preferring the LCP-array from Cor.\ \ref{cor:wee_lcp1} over that of Cor.\ \ref{cor:wee_lcp2}.

\section*{Acknowledgments}
The author wishes to express his gratitude towards the anonymous reviewers, whose insightful comments helped to improve the present material substantially.

\bibliographystyle{abbrv}
\bibliography{wee}

\end{document}